\documentclass{article}
\usepackage{fullpage}
\usepackage{amsmath,amsfonts,amssymb,dsfont,bm}
\usepackage{bbm}
\usepackage{graphicx, listings, cite, hyperref}

\usepackage{algorithm2e}

\usepackage{comment}
\newcommand{\eps}{\epsilon}

\usepackage{geometry}
\usepackage{lmodern}

\usepackage{graphicx}
\usepackage{lscape}
\usepackage[final]{pdfpages}
\usepackage{amsthm}
\usepackage{amsmath,bm}
  {
      \theoremstyle{plain}
      
      \newtheorem{theorem}{Theorem}[section]
      \newtheorem{lemma}[theorem]{Lemma}
      \newtheorem{definition}[theorem]{Definition}

     \newtheorem{conjecture}[theorem]{Conjecture}

  }
\usepackage{amsfonts}
\usepackage{mathtools}
\usepackage[]{algorithm2e}
\usepackage{color}
\DeclareMathOperator{\sign}{sign}

\author{Alexandr Andoni\\CS and DSI\\Columbia University \and David Cheikhi\\Columbia University}

\title{From Average Embeddings To Nearest Neighbor Search\footnote{Research supported in part by NSF grants (CCF-2008733, CCF-1617955, and CCF-1740833), and Simons Foundation (\#491119).}}
\begin{document}

\maketitle


\begin{abstract}
    In this note, we show that one can use average embeddings, introduced recently in \cite{naor2020average}, to obtain efficient algorithms for approximate nearest neighbor search. In particular, a metric $X$ embeds into $\ell_2$ on average, with distortion $D$, if, for any distribution $\mu$ on $X$, the embedding is $D$ Lipschitz and the (square of) distance does not decrease on average (wrt $\mu$). In particular existence of such an embedding (assuming it is efficient) implies a $O(D^3)$ approximate nearest neigbor search under $X$. This can be seen as a strengthening of the classic (bi-Lipschitz) embedding approach to nearest neighbor search, and is another application of data-dependent hashing paradigm.
\end{abstract}

\section{Introduction}
In the $c$-Approximate Nearest Neighbor Search ($c$-ANN) problem, we preprocess a dataset $P$ living in some metric space $(X,d_X)$ for some threshold $r>0$,  such that, for a given query $q\in X$, we can efficiently find a point in $p\in P$ at distance $d_X(p,q)\le cr$, as long as there exists a point $p^*\in P$ at distance $d_X(q,p^*)\le r$. In many applications, the metric space is high-dimensional, with the $d$-dimensional Euclidean space $\Re^d$ being the most well-studied. See, e.g., survey \cite{andoni18-icm} for further background on high-dimensional ANN problem.

A common approach for solving the ANN problem is via 
Locality Sensitive Hashing (LSH) \cite{HIM12}, which can be viewed as a random, oblivious space partition. LSH mainly applies to the Euclidean ($\ell_2$) and Hamming ($\ell_1$) spaces. For example, for any fixed $\eps>0$, one can solve ANN under $
\ell_2$ with $c=\sqrt{1/\eps}+o(1)$ approximation, $O(n^\eps d)$ query time and $O(n^{1+\eps}d)$ space (we call such query/space parameters as "efficient" henceforth).

Beyond $\ell_1$ and $\ell_2$ metrics, there has been much less
progress and the best bounds are usually far from understood. A
long-standing approach is via {\em metric embeddings}: a map from a metric $(X,d_X)$ into an "easier" space, say, $(\Re^d,\ell_2)$, while approximately preserving all the distances:

\begin{definition}[bi-Lipschitz embedding into $\ell_2$]
A map $f:X\to \ell_2$ is an embedding  into $\ell_2$, with distortion $D\ge 1$, iff for any $x,y\in X$ we have that
$$
\|f(x)-f(y)\|_2\le d_X(x,y)\le D\cdot \|f(x)-f(y)\|_2.
$$
\end{definition}

Using such an embedding for a metric $(X,d_X)$, one can immediately extend the ANN under $\ell_2$ algorithms to ANN under $X$, with an extra factor $D$ approximation. Indeed, this approach yields best-known ANN algorithms for spaces such as the edit distance \cite{OR-edit}. However, there are also known lower bounds for the best possible distortion, including in the case of edit distance \cite{KN, KR06}. A possible extension of the embedding approach is to embed a {\em snowflake} of the metric $(X,d_X)$ which means embedding $(X,d_X^{\omega})$ for some $\omega<1$, though oftentimes there are lower bounds for this approach as well (albeit weaker), see, e.g., \cite{AK07} for our example of edit distance.

A recent qualitative development is that of {\em data-dependent} hashing/space partitioning methods.  For $\ell_2$ and $\ell_1$ spaces, \cite{AINR-subLSH, AR-optimal, alrw17-tradeoffs} developed data-dependent LSH that gave a better approximation: e.g., $c=\sqrt{\tfrac{1+\eps}{2\eps}}+o(1)$ for $\ell_2$; also see~\cite{Ind98} for an earlier example of data-dependent technique for $\ell_\infty$.

For metrics beyond $\ell_1, \ell_2$, \cite{annrw18-spectralGaps} developed a new data-dependent approach, based on the following geometric property of the space $(X,d_X)$, termed cutting modulus:

\begin{definition}[Cutting modulus]
For fixed $\eps>0$, the {\em cutting modulus} of the metric ${\cal X}=(X,d_X)$, denoted {\em $\Xi({\cal X},\eps)$} is the infimum $D$ such that the following holds for any $r>0$ and any integer $m>1$.  Fix a pointset $V=\{x_1,\ldots x_m\}\subset X$, and consider a positively-weighted graph $G=(V,E,w)$ where $w_{pq}>0 \implies d_X(p,q)\le r$ and $\sum_{p\in V}\sum_{q\in V} w_{p,q}=1$ (i.e., it's a distribution over the ordered pairs $(p,q)$). Then, one of the following must hold:
\begin{itemize}
    \item Either there exists some point $z\in X$ such that $\sum_{p:d_X(z,p)\le Dr} \sum_{q\in V} w_{p,q}\ge 1/2$ (i.e., there is a ball of radius $Dr$ containing 1/2 mass of the pointset),
    \item Or there exists a non-trivial cut $S\subset V$ with conductance $\Phi(S)\le \eps$.
\end{itemize}
\end{definition}

The main result of \cite{annrw18-spectralGaps} is that any metric $X$ with cutting modulus $D$ admits a $D$-ANN, albeit with a couple important caveats. First, without further properties, the query algorithm is efficient only in the cell-probe model, as it involves operating with an exponentially-sized graph $G$ (think a net in $\cal X$). Nonetheless, as \cite{annrw18-spectralGaps} show, one can use a slightly stronger notion of cutting modulus to remove this caveat, for example obtaining efficient ANN for $\ell_p$ and Schatten-$p$ norms, with $O(p/\eps)$ approximation, for any $p\ge 2$. Second, the preprocessing time is exponential and seems unavoidable for this approach.

There are currently few bounds on the cutting modulus, and most of them follow from the theory of {\em non-linear spectral gaps}, whose systematic study was started in \cite{mendel2014nonlinear}; see also \cite{naor2014comparison,mendel2015expanders}. Perhaps most striking is that the cutting modulus of any $d$-dimensional norm $(\Re^d,\|\cdot\|_X)$ is $O(\tfrac{\log d}{\eps})$, following the bound on the non-linear spectral gap from \cite{naor16-spectralGap}. 

The work of \cite{naor2020average} introduced the notion of {\em average embedding} and proved a connection to non-linear spectral gap bounds. An average embedding relaxes the bi-Lipschitz embedding from above in that the contraction holds only on average with respect to a fixed given measure (and hence the embedding is "data dependent"):

\begin{definition}[Average embedding] For a metric ${\cal X}=(X,d_X)$, a map $f:X \rightarrow \ell_2$ is an {\em $q$-average embedding}, for a dataset $P\subset X$, with distortion $D$ if the following is satisfied:
\begin{align*}
    \|f(p) - f(q)\|_2 &\leq D\cdot d_X(p,q) \quad \text{ for all $p,q \in X$}\\
     \sum_{p,q \in P} \|f(p)-f(q)\|_2^q&\ge \sum_{p,q \in P} d_X(p,q)^q. 
\end{align*}
When $q=2$, we call it simply {\em average embedding}.
\end{definition}

In particular,  \cite{naor2020average} showed that the average embedding distortion for $1/2$-snowflake of a $d$-dimensional normed space is $O(\sqrt{\log d})$. Similarly, \cite{naor2014comparison, naor2020average} show that the average embedding distortion for $\ell_p$ is $O(p)$, for $p>2$. Both of these embeddings are non-constructive, proven by duality.

\paragraph{Our result.} In this note, we show that one can use average
embeddings directly to solve ANN.

\begin{theorem}
Suppose a metric $X$ has an average embedding $f$ with distortion $\sqrt{D}$, where time to construct $f$ is $T_P$ and time to compute $f$ on a point $q$ is $T_c$. Then there is a $c$-ANN with query time $O(n^{\eps}\cdot T_c)$, space $n^{O(1)}$, and preprocessing time $n^{O(1)}\cdot T_P$, and approximation $c=O(D^{3}/\epsilon)$.
\end{theorem}

Our approach avoids some of the aforementioned caveats of using cutting modulus, as long as we can construct an efficient average embedding for a given metric $X$. In particular, there is no more need for exponential time preprocessing (again assuming efficient average embedding map). 

At high-level, our algorithm is similar to the ANN approach via standard bi-Lipschitz embedding: embed the space $X$ into $\ell_2$ and then use an efficient ANN for $\ell_2$, though the actual algorithm is a little bit more nuanced. The overall algorithm is "data-dependent" because the average embedding $f$ is data-dependent (note that $f$ from above may depend on the dataset $P$). In Appendix \ref{apx:ell_p} we also propose an explicit, efficient average embedding for $\ell_p$, $p>2$, albeit we leave it as a conjecture to prove its second property.

\paragraph{Independent work.} In independent work, \cite{knt21-average} showed a similar result, obtaining a tighter bound of $O(D\log D)$-ANN. Their algorithm requires one more step: to construct a weak average embedding out of a generic average embedding. They also show how to compute the $1$-average embedding for $\ell_p$ with distortion $O(p)$ (thus also resolving a related version of our conjecture), and hence obtain an $O(p)$-ANN with polynomial-time preprocessing (removing also the $O(\log p)$ from their generic reduction). Their algorithm thus improves over ANN from \cite{DBLP:journals/tcs/BartalG19,annrw18-spectralGaps}.

\section{ANN Algorithm from Average Embedding}

We now describe our main ANN algorithm that uses an average embedding. Our algorithm is solves a special case of the ANN problem, termed {\em bounded ANN}: 
the {\em $\beta$-bounded ANN} is the problem where we are  guaranteed that all distances within the dataset $P$ verify $r \leq d_X(p,q) \leq \beta\cdot cr$. 
We note that it is enough to solve $\beta$-bounded ANN problem, for $\beta=18$, due to the result of \cite{DBLP:journals/tcs/BartalG19}. 

Our algorithm uses an LSH scheme for $\ell_2$, formally defined as follows.

\begin{definition}(LSH)
For a metric $X$, a family of hash functions $h:X\to U$ is $(p_1,p_2,r,cr)$-LSH if, for all $x,y\in X$: \begin{itemize}
    \item if $d_X(x,y)\le r$, then
$\Pr[h(x)=h(y)]\ge p_1$,
\item if $d_X(x,y)> cr$, then $\Pr[h(x)=h(y)]< p_2$.
\end{itemize}
The exponent of the LSH scheme is defined as $\rho=\tfrac{\log 1/p_1}{\log 1/p_2}$.
\end{definition}

The metric $X=\ell_2$ admits $(p_1,p_2,r,cr)$-LSH, for $\rho=1/c$ with $p_1$ arbitrarily close to 1 \cite{DIIM}. We also assume that $r=1$ henceforth for simplicity.

Fix some $\lambda, w\ge1$ tbd. We construct the data structure for a $\beta$-bounded pointset $P$ recursively as a collection of randomized trees, constructed iid. A tree  recursively partitions the dataset until the size of the dataset  becomes constant (below think of $P$ is the dataset assigned to the current node). There are two type of nodes:
\begin{itemize}
    \item If there exists $x_0 \in P$ such that $|P \cap B_X(x_0,
      \lambda D))| > |P|/8$, we store $x_0$ and a point $p_0 \in P
      \cap B_X(x_0, \lambda D)$. For the rest of the points, we create a single child inheriting $P \setminus B_X(x_0, \lambda D))$, and recurse the construction on it.
    
    \item If there is no $x_0 \in P$ such that $|P \cap B_X(x_0, \lambda D)| > |P|/8$, then we compute an average embedding $f:X\to \ell_2$ for $P$. Then we sample a $(p_1,p_2,D,wD)$-LSH $h:\ell_2\to U$ and use it to partition the pointset $P$. In particular, each part $k\in U$ generates a child node $k$ that inherits the points $p\in P$ with $h(f(p))=k$ (storing only non-empty parts as usual).
\end{itemize}
For a given query point $q$, we traverse the tree as follows:
\begin{itemize}
    \item if the node is of the first type: if $d(x_0, q) < \sqrt{D}$, return the associated $p_0$, and otherwise recurse into the single child;
    \item if the node is of the second type: branch into the child node $k=h(f(q))$.
    \item in a leaf node: check the distance to all stored nodes and return a valid answer if there exists one.
\end{itemize}

\begin{theorem}
\label{thm:oneTree}
For any $\epsilon>0$, there is $\lambda, w$, and $c=O(D^{3}/\epsilon)$ such that
the tree constructed as above has a probability at least $n^{-\epsilon}$ to find a $cr$-near neighbor, assuming there's an $r$-near neighbor.
\end{theorem}

Hence, we can build a $\beta$-bounded ANN which constructs $n^{\epsilon}$ trees as above, and uses 
space $n^{1+O(\epsilon)}$ for approximation $c=O(D^{3}/\epsilon)$. The reduction from \cite{DBLP:journals/tcs/BartalG19} then yields an algorithm for the vanilla ANN.

\subsection{Analysis: Proof of Theorem~\ref{thm:oneTree}}

First we note that nodes of the first type do not affect the correctness. So in the rest we fix a node of the tree above, and suppose we are in the second case: there is no $x_0 \in X$ such that $|P \cap B_X(x_0, \lambda D)| > |P|/8$. We further assume that $P$ contains a near neighbor point $p^*$, with $d(q,p^*)\le1$. We first prove the following lemma bounding the number of points that are "close" to $q$ after the embedding $f$.

\begin{lemma}
Consider a set $P$ satisfying:
\begin{itemize}
    \item $|B_X(x, \lambda D) \cap P | \leq p |P|$ for all $x \in P$;
    \item For all points $x, y \in P$, $d_X(x,y) \leq \beta c$.
\end{itemize}
Consider a query $q$ at distance at most 1 from some point in $P$.
Then, we have that $$\left|B_{\ell_2}(f(q),wD) \cap f(P)\right| \leq  1 - \dfrac{(1-p)\lambda^2}{4\beta^2c^2},$$ as long as $(1-p)\lambda^2=\Omega(w^2)$, and $\beta c\ge \lambda$.
\end{lemma}

\begin{proof}

Let $\alpha = \tfrac{\left|B_{\ell_2}(f(q), wD) \cap f(P)\right|}{|P|}$. By the second property of the average embedding for $P$:
\begin{align*}
    \sum_{x,y \in \mathcal{P}} d(x,y)^2 &\leq  \sum_{x,y \in \mathcal{P}} \|f(x)-f(y)\|^2\\ 
    \implies 
    \sum_{x,y \in \mathcal{P}} d(x,y)^2 &\leq 
     (1-\alpha^2)\beta^2c^2D^2n^2+4\alpha^2 w^2D^2n^2\\
    \intertext{using $D$-Lipschitzness of $f$, and that for an $\alpha$ ratio of $x\in P$, $\|f(x)-f(q)\| \leq wD$. }
    \implies   n^2(1-p)\lambda^2D^2 &\leq 
    \beta^2c^2D^2n^2-\alpha^2 (\beta^2c^2D^2-4w^2D^2)n^2
    \\
    \intertext{since $d_X(x,y) \geq \lambda D$ for a $1-p$ ratio of $x,y \in P$.}
    \implies  
    \alpha^2 &\le \frac{\beta^2 c^2D^2-(1-p)\lambda^2D^2}{\beta^2c^2D^2-4w^2D^2}
    \\
    \implies  
    \alpha^2 &\le \frac{\beta^2 c^2-(1-p)\lambda^2}{\beta^2c^2-4w^2}
    \\
    \implies  
    \alpha^2 &\le (1-\tfrac{(1-p)\lambda^2}{\beta^2c^2})(1+O(w^2/\beta^2c^2))
    \\
    \implies  
    \alpha^2 &\le 1-\tfrac{(1-p)\lambda^2-O(w^2)}{\beta^2c^2}.
    \\
    \implies  
    \alpha &\le 1-\tfrac{(1-p)\lambda^2}{4\beta^2c^2},
\end{align*}
for $(1-p)\lambda^2\ge \tfrac{1}{2}O(w^2)$.
\end{proof}

We use the lemma from above with $p=1/8$, $\beta=18$, and $c=\lambda D$.
After the embedding $f$, it is guaranteed that $|f(P) \cap B(f(q), w D)| > |P|\cdot(1 - \Omega\left(\dfrac{1}{D^2}\right))$. Using an $(p_1,p_2,D,wD)$-LSH in the embedded space, the fraction of points that collides with the query is:
\begin{equation*}
    p_2' \le 1 - \Omega\left(\dfrac{1}{D^2}\right) + p_2 \left(1-(1 - \Omega\left(\dfrac{1}{D^2}\right))\right) \le 1 - (1-p_2)\dfrac{\Omega(1)}{D^2}.
\end{equation*}

Therefore, embedding the query and the dataset and applying an $(p_1,p_2,D,wD)$-LSH is like applying an $(p_1,p_2',1,\lambda D)$-LSH in the original space. Therefore, we pick
\begin{equation*}
    p_1 = 1 - \epsilon (1-p_2') = 1 -\epsilon (1-p_2)\dfrac{\Theta(1)}{D^2}. 
\end{equation*}

Such a choice is possible as long as $w\ge \Theta(D^2/\epsilon)$. As, from above, we require $\lambda \geq \Omega(w)$, we have that $\lambda \geq \Omega(D^2/\epsilon)$. Finally, we get that the approximation ratio is $c=\lambda D=O(D^3/\epsilon)$. We obtain $\rho=\eps$ and hence $n^{-\eps}$ success probability.





\bibliographystyle{alpha}
\bibliography{bib,andoni}

\appendix
\section{Average Embedding for $\ell_p$}
\label{apx:ell_p}

We consider the following average embedding from $\ell_p$ to $\ell_2$.
Let \begin{equation*}
    h(x)_i = \sign(x_i)|x_i|^{p/2} 
\end{equation*}
and \begin{equation*}
    f(x) = \dfrac{h(x)}{\|h(x)\|_2}\|x\|_p = h(x) \dfrac{\|x\|_p}{\|x\|_p^{p/2}} = h\left(\dfrac{x}{\|x\|_p}\right)\|x\|_2
\end{equation*}

\begin{theorem}
The embedding $f: \ell_p \rightarrow \ell_2$ is $p+1$ Lipschitz.

\end{theorem}
\begin{proof}
We notice that $h$ is the Mazur map from $\ell_p$ to $\ell_2$. The embedding $f$ first normalizes its input, applies the Mazur map and the rescale it. We will therefore make use of the fact that the Mazur map is $\frac{p}{2}$-Lipscitz on the unit sphere. 

\begin{align*}
    \|f(x)-f(y)\|_2 &= \left\|\dfrac{h(x)}{\|h(x)\|_2}\|x\|_p - \dfrac{h(y)}{\|h(y)\|_2}\|y\|_p\right\|_2\\
   & = \left\|\|x\|_p\left(\dfrac{h(x)}{\|h(x)\|_2} - \dfrac{h(y)}{\|h(y)\|_2}\right)+ \dfrac{h(y)}{\|h(y)\|_2}\left(\|x\|_p-\|y\|_p\right)\right\|_2\\
   &\leq \|x\|_p \left\|\dfrac{h(x)}{\|h(x)\|_2} - \dfrac{h(y)}{\|h(y)\|_2}\right\|_2+ \left\|\dfrac{h(y)}{\|h(y)\|_2}\right\|_2\left|\|x\|_p-\|y\|_p\right|\\
   &\leq \|x\|_p \left\|\dfrac{h(x)}{\|h(x)\|_2} - \dfrac{h(y)}{\|h(y)\|_2}\right\|_2+ \left\|x-y\right\|_p\\
   &\stackrel{(a)}{=} \|x\|_p \left\|h\left(\dfrac{x}{\|x\|_p}\right) - h\left(\dfrac{y}{\|y\|_p}\right)\right\|_2+ \left\|x-y\right\|_p\\
   &\stackrel{(b)}{\leq}\dfrac{p}{2}\|x\|_p\left\|\dfrac{x}{\|x\|_p}-\dfrac{y}{\|y\|_p}\right\|_p +  \left\|x-y\right\|_p\\
   &= \dfrac{p}{2}\left\|x-\dfrac{\|x\|_p}{\|y\|_p}y\right\|_p +  \left\|x-y\right\|_p\\
   &= \dfrac{p}{2}\|x-y\|+\dfrac{p}{2}\left|\dfrac{\|x\|_p}{\|y\|_p}-1\right|\|y\|_p +  \left\|x-y\right\|_p\\
   &\leq(p+1)\left\|x-y\right\|_p
\end{align*}

Where (a) uses the fact that $h(x)/\|h(x)\|_2 = h(x/\|x\|_p)$ and (b) uses the fact that $h$ is $p/2$-Lipschitz on the unit sphere. 
\end{proof}

To complete the proof that the map $f$ is average embedding we need the following inequality, left as a conjecture.

\begin{conjecture}
For any $x_1, \dots, x_n$, we can find $z$ such that $x\rightarrow f(x-z)$ verifies, for some universal constant $C>0$:
\begin{equation*}
    \sum_{i,j} \|f(x_i - z) - f(x_j -z)\|_2^2 \geq C \sum_{i,j} \|x_i - x_j\|_p^2.
\end{equation*}
\end{conjecture}

\end{document}